\documentclass[onecolumn, conference]{IEEEtran}
\usepackage[utf8]{inputenc} 
\usepackage[T1]{fontenc}
\usepackage{url}
\usepackage{ifthen}
\usepackage{cite}
\usepackage[cmex10]{amsmath} % Use the [cmex10] option to ensure complicance
\usepackage{graphicx} % Required for inserting images
\usepackage{xfrac}

\usepackage{amsthm}
\theoremstyle{definition}
\newtheorem{theorem}{Theorem}%[section]
\newtheorem{proposition}[theorem]{Proposition}

\newtheorem{corollary}[theorem]{Corollary}
\theoremstyle{definition}
\newtheorem{definition}[theorem]{Definition}

\theoremstyle{remark}

\usepackage{amssymb}
\DeclareMathOperator*{\E}{\mathbf{E}}
\DeclareMathOperator*{\Var}{\mathbf{Var}}
\DeclareMathOperator{\D}{\mathbf{D}_{KL}}
\DeclareMathOperator*{\argmin}{arg\,min}
\newcommand{\Eb}[1]{\E\left[#1\right]}
\newcommand{\Varb}[1]{\Var\left[#1\right]}
\newcommand{\R}{\mathbb{R}}
\newcommand{\Z}{\mathbb{Z}}
\newcommand{\N}{\mathbb{N}}
\newcommand{\pr}[1]{\mathbf{P}\left\{#1\right\}}
\newcommand{\unifint}[1]{\text{Unif}\left[#1\right)}
\newcommand{\unifset}[1]{\text{Unif}\left[#1\right]}

\newcommand{\ucell}{V}
\newcommand{\cell}{\mathcal{C}}
\newcommand{\norm}[1]{\left\lVert#1\right\rVert_2}
\newcommand{\normsqr}[1]{\left\lVert#1\right\rVert_2^2}

\newcommand{\Gm}[1]{\Gamma\left(#1\right)}

\newcommand{\gtwo}{\gamma_2}
\newcommand{\gfour}{\gamma_4}

\newcommand{\lattice}{\Lambda}
\newcommand{\fnice}{\bar{f}}
\newcommand{\cnst}{C}

\usepackage{xcolor}
\usepackage{tikz}

\begin{document}
\title{Gaussian Channel Simulation with \\ Rotated Dithered Quantization} 

\author{%
  \IEEEauthorblockN{Szymon Kobus\IEEEauthorrefmark{1},
                    Lucas Theis\IEEEauthorrefmark{2},
                    Deniz G\"und\"uz\IEEEauthorrefmark{1}}
  \IEEEauthorblockA{\IEEEauthorrefmark{1}%
                    Department of Electrical and Electronics Engineering, Imperial College London, London, UK}
  \IEEEauthorblockA{\IEEEauthorrefmark{2}%
                    Google DeepMind, London, UK}
  \IEEEauthorblockA{\{szymon.kobus17, d.gunduz\}@imperial.ac.uk, theis@google.com}
}
\maketitle
\begin{tikzpicture}[remember picture,overlay]
		\node[anchor=north,yshift=-10pt] at (current page.north) {\parbox{\dimexpr\textwidth-\fboxsep-\fboxrule\relax}{
				\centering\footnotesize This paper has been accepted for presentation at the 2024 IEEE International Symposium on Information Theory. \textcopyright 2024 IEEE.
				}};
\end{tikzpicture}

\begin{abstract}
    Channel simulation involves generating a sample $Y$ from the conditional distribution $P_{Y|X}$, where $X$ is a remote realization sampled from $P_X$. This paper introduces a novel approach to approximate Gaussian channel simulation using dithered quantization. Our method concurrently simulates $n$ channels, reducing the upper bound on the excess information by half compared to one-dimensional methods. When used with higher-dimensional lattices, our approach achieves up to six times reduction on the upper bound. Furthermore, we demonstrate that the KL divergence between the distributions of the simulated and Gaussian channels decreases with the number of dimensions at a rate of $O(n^{-1})$.
\end{abstract}

\section{Introduction}
Channel simulation concerns sampling from a distribution $Y\sim P_{Y|X}$ at the decoder, given a realization of $X\sim P_X$ at the encoder, while transmitting the least number of bits.  It can be used in end-to-end learned lossy compression, where channel simulation can replace the non-differentiable quantization step and remove the mismatch between the soft quantization employed during training and hard quantization during inference \cite{flamich_compressing_2020, agustsson_universally_2020, theis_lossy_2022, havasi_minimal_2019}. 
It has been recently used in federated learning \cite{isik_communication-efficient_2023} to increase the communication efficiency, where the clients simply enable the parameter server to sample from their updated model distributions, rather than sending exact model updates. It also finds application in differential privacy, where the simulated channel corresponds to the desired noise distribution \cite{triastcyn_dp-rec_2021, pmlr-v151-shah22b, burak_comm_2024}.

Asymptotic results for channel simulation were studied in \cite{cuff_communication_2008}.
% , where it is shown that the minimal required rate is given by Wyner's common information $C(X;Y)$ \cite{wyner_common_1975}. 
In the presence of common randomness between the encoder and the decoder the rate reduces to mutual information $I(X;Y)$, which was discovered earlier in the context of quantum information theory in \cite{bennett_entanglement-assisted_2002, winter_compression_2002}. For the one-shot setting, it is shown in \cite{li_strong_2018} that for any $X,Y$, channel simulation can be performed within $I(X;Y)+\log\left(I(X;Y)+1\right)+5~\mathrm{bits}$. This is achieved by generating a list of candidate samples from prior $P_Y$ at both the encoder and the decoder. Using $P_{Y|X}$ the encoder communicates an index from this list, and the corresponding sample is recovered by the decoder, which is guaranteed to follow $P_{Y|X}$.
However, the number of samples that need to be generated (i.e., the computational complexity) of this method is proportional to $\exp(D_\infty({P_{Y|X}||{P_Y}}))$ \cite{maddison2016ppmc}. Importance sampling \cite{harsha_communication_2010} is an approximate method of channel simulation in which the number of drawn samples can be exchanged for the quality of obtained samples \cite{havasi_minimal_2019}, but it results in higher coding costs. Ordered random coding \cite{theis_algorithms_2022} is a synthesis of the two approaches; it mimics the coding cost of Poisson functional representation and computation control of importance sampling. However, for the guarantees of importance sampling and ordered random coding to hold, the number of drawn samples needs to be at least $\exp(\D(P_{X|Y}||P_X)+t), t>0$. 

The computational complexity is a major hurdle for channel simulation in practical applications. It is shown in \cite{agustsson_universally_2020} that there is no general channel simulation algorithm with a computational cost that is polynomial in the information content. On the other hand, the computational complexity can be low for specific distributions. Subtractive dithered quantization \cite{ziv_universal_quantization, zamir_universal_1992} is a method to simulate an additive uniform noise channel $Y=X+U$, where $U$ is uniformly distributed over $[-\sfrac{1}{2},\sfrac{1}{2})$. Unlike the aforementioned approaches, it requires only a single common sample. Simulating a scalar Gaussian channel using one-dimensional dithered quantization is considered in \cite{agustsson_universally_2020} using a randomized scale variable, which follows from the scale mixture of uniform distributions \cite{CHOY2003227}. This approach was further extended in \cite{hegazy_randomized_2022} to simulate any unimodal additive noise distribution with randomized scaling and offsets.

This work focuses on using dithered quantization in $n$ dimensions to simulate an additive Gaussian channel. In particular, we are interested in the simulation of $Y=X+C$, where $C \in \R^n$ is independent of $X$ and follows a multivariate Gaussian distribution.
The scale mixture of uniform distributions can be generalized to multi-variate exponential power distributions using a uniform distribution over a unit ball in $n$ dimensions \cite{FUNG20082883}; however, the error distribution of dithered quantization using $n$-dimensional lattices is uniform not over the unit ball, but over the Voronoi cell of the lattice. A complicated vector quantization scheme that results in an error uniformly distributed over the unit ball is proposed in \cite{Ling:ITW:23}. In this paper, we take an alternative approach and consider instead approximate channel simulation using dithered quantization. Our results show that using higher-dimensional lattice quantizers leads to significantly reduced coding overhead. Additionally, as our scheme is based on dithered quantization, it inherits its low computational complexity.

The contributions of this work are summarized as follows:
\begin{itemize}
    \item We introduce a novel and computationally efficient method of approximate Gaussian channel simulation based on subtractive dithered quantization.
    \item We show that the KL divergence between the error distribution $Y-X$ induced by the proposed method and a Gaussian multi-variate distribution scales as $O(n^{-1})$.
    \item We upper bound the required number of bits of the proposed method for different lattice quantizers, showing two-fold decrease for an integer lattice and up to six-fold reduction for higher-dimensional lattices compared to layered quantization in \cite{hegazy_randomized_2022}.
\end{itemize}

{\textbf{Notation.}} For a vector $x = (x_1, \ldots, x_n) \in\R^n$, we have $x_{a:b}=(x_a, x_{a+1},\dots x_b)^\top$ and $x_{:b}=x_{1:b}$. The $\ell_2$ norm of $x$ is denoted $\norm{x}=\sqrt{\sum_{i=1}^n x_i^2}$, $\Gm{\cdot}$ stands for Gamma function, $\lfloor\cdot\rceil$ for rounding to the nearest integer, while $h(X)$ denotes the differential entropy of random variable (r.v.) $X$.

\section{Dithered Quantization}
Subtractive dithered quantization is a method of channel simulation, where noise $C$ is uniform over the Voronoi cell of a lattice. For the one-dimensional case, we have $C\sim \unifint{-\sfrac{1}{2}, \sfrac{1}{2}}$. We assume the encoder and decoder have access to a source of common randomness independent of $X$, using which they can draw samples from a uniform r.v. $U'\sim\unifint{-\sfrac{1}{2}, \sfrac{1}{2}}$. The encoder rounds $X-U'$ to the nearest integer value and communicates it to the decoder, which in turn adds the same realisation of $U'$ to it, obtaining:
\begin{equation}
    \lfloor X - U' \rceil + U' \sim X + U,
\end{equation}
where $U\sim\unifint{-\sfrac{1}{2}, \sfrac{1}{2}}$. It is shown in \cite{ziv_universal_quantization} that this process is efficient, in terms of mean squared error, as a quantization method even for high dimensional sources. 

\begin{definition}[Lattice]
    Given basis vectors $e_1, e_2, \dots, e_n$, an $n$-dimensional lattice $\lattice$ is all their integral combinations:
    \begin{equation}
        \lattice = \left\{\sum_{i=1}^n k_i e_i: k_i\in\Z\right\}.
    \end{equation}
    The quantization function of lattice $\lattice$ is defined as:
    \begin{equation}
        q_{\lattice}(x)\triangleq \argmin_{e\in\lattice}\norm{x-e},
    \end{equation}
    and the Voronoi cell of lattice $\lattice$ as:
    \begin{equation}
        \cell_\lattice\triangleq\left\{x\in\R^n\,:\,q_\lattice(x)=0\right\}.
    \end{equation}
\end{definition}
For one-dimensional dithering, rounding $\lfloor\cdot\rceil$ can be interpreted as quantization to the closest point in the one-dimensional lattice $\Z$.
We can generalize this method to higher-dimensional lattices. The uniform noise $U'$ is replaced by uniform samples from the Voronoi cell $\cell_\lattice$, i.e., $V'\in\R^n,V'\sim\unifset{\cell_\lattice}$, and rounding is replaced by quantization to the closest lattice point $q_{\lattice}$. For $X\in\R^n$, this procedure simulates an $n$-dimensional channel with uniform noise over the Voronoi cell of the lattice $\cell_\lattice$:
\begin{equation}
     q_\lattice(X - V') + V' \sim X + V,
\end{equation}
where $V\sim\unifset{\cell_\lattice}$.

If we apply dithered quantization to $M^{-1}X$, with a fixed invertible matrix $M\in \R^{n\times n}$, the decoder would obtain $M^{-1}X + V$. If $M$ is known at the decoder, it can recover $Y=M(M^{-1}X+V)=X+MV$ with quantization error $MV$.

A special orthogonal group $SO(n)$ represents the set of all possible rotations in $n$ dimensions.
Let $R\in \R^{n\times n}$ be a random rotation matrix distributed according to the Haar measure over $SO(n)$ (i.e., uniform rotation). Let $R$ be generated with common randomness, and be available at both the encoder and decoder.
Then, the dithering procedure can be applied to $R^\top X$ (since $R^{-1}=R^\top$) obtaining quantization error $RV$. For a given realization of $R$, $P_{RV|R}$ is uniform over the Voronoi cell of the lattice rotated by $R$, since $\det(R)=1$. This procedure can be thought of as applying dithered quantization to $X$ using a randomly rotated lattice, and is captured by the equation:
% X -> R^\top X -> R^\top X - V -> q(R^\top X - V) + V -> R^\top X + V' -> X + RV'
\begin{equation}
    R (q_\lattice(R^\top X - V') + V') \sim X + RV~.
\end{equation}

Unlike $P_{RV|R}$, $P_{RV}$ is not uniform, and has a rotationally invariant density function over an $n$-dimensional ball. A function $f:\R^n\to\R$ is rotationally invariant if all $a,b\in\R^n$ with $\norm{a}=\norm{b}$ $f(a)=f(b)$ holds. Hence, $f$ can be equivalently described with a function $g:\R\to\R$, where $f(a)=g(\norm{a})$. 
\begin{proposition} \label{proposition:R_makes_rot_invariant}
    Let $V\in\R^n$ be a r.v. with distribution $P_V$, and $R\in\R^{n\times n}$ be drawn from $SO(n)$ according to the Haar measure. Then, the probability density function of $RV$ is rotationally invariant.
\end{proposition}
\begin{proof}
    Let $M\in \R^{n\times n}, MM^\top=I, \det(M)=1$ be a fixed rotation matrix, then  $MRV$ and $RV$ follow the same distribution, since $R$ is uniformly distributed over rotation matrices, so is $MR$ (Haar measure is invariant under left multiplication). For all $a,b\in\R^n$ with $\norm{a}=\norm{b}$ there exists a rotation matrix $M$ such that $Ma=b$.
    Therefore,
    \begin{align}
        p_{RV}(a)=p_{RV}(M^\top b)=p_{MRV}(b)=p_{RV}(b).
    \end{align}
\end{proof}
 
\section{Outline of rotated dithered quantization}
The goal of this work is to simulate an $n$-dimensional channel with Gaussian noise. To simulate a Gaussian channel with known covariance $\Sigma=AA^T$ it suffices to simulate one with identity covariance and apply it to an input sample $A^{-1} X$ at the encoder. The decoder can multiply the output $A^{-1} X + G, G\sim \mathcal{N}(0, I_n)$ by $A$, to obtain $X + AG, AG\sim\mathcal{N}(0,\Sigma)$. Thus, we focus on simulating an $n$-dimensional Gaussian channel with identity covariance.

The probability density function of a zero-mean Gaussian distributed $G$ with an identity covariance matrix, i.e., $G\sim \mathcal{N}(0, I_n)$, is rotationally invariant, and can be fully described by the distrobution $P_{\norm{G}}$ of the $\ell_2$-norm of $G$. Using dithered quantization we can achieve uniform error distribution $\ucell$ over the Voronoi cell of a lattice. Combining dithered quantization with random rotation $R$ the probability density function of error distribution becomes $RV$, which is also rotationally invariant. 
Subsequently, the task simplifies to `matching' the distributions of $\norm{G}$ and $\norm{RV}$. This will be achieved by scaling and perturbing the reconstruction.

For a rotationally invariant random variable $X$, it is enough to characterize the distribution of $\normsqr{X}$ as we can obtain that of $\norm{X}$ through a deterministic transformation $p_{\norm{X}}(t)=2t\;p_{\normsqr{X}}(t^2)$. For $G\sim \mathcal{N}(0, I_n)$, where $G_i\sim \mathcal{N}(0,1)$, we have
\begin{equation}
    \normsqr{G}=\sum_{i=1}^n G_i^2 \stackrel{\text{approx.}}{\sim} \mathcal{N}(\mu=n, \sigma^2=2n), \label{eqn:l2_gaussian}
\end{equation}
follows a chi-squared distribution $\normsqr{G}\sim\chi^2_n$. By the central limit theorem, $\normsqr{G}$ converges to a Gaussian distribution as $n$ increases. %The quality of this approximation is analyzed in Section \ref{section:kl_analysis}.

\subsection{Integer lattice}
One of the simplest lattices is an integer lattice $\Z^n$, the Voronoi cell of which is a unit $n$-dimensional cube. Dithered quantization with such a lattice is equivalent to performing one-dimensional dithering for each dimension. Denoting the error of quantization with this lattice by $U$, where $U_i\sim\unifint{-\sfrac{1}{2}, \sfrac{1}{2}}$, we have
\begin{equation}
    ||U||_2^2=\sum_{i=1}^n U_i^2 \stackrel{\text{approx.}}{\sim} \mathcal{N}\left(\frac{n}{12}, \frac{n}{180}\right).
\end{equation}
% since 
% \begin{align*}
%     \E U_i^2 &= \int_{-\frac{1}{2}}^{\frac{1}{2}} x^2 dx = \Big[\frac{1}{3}x^3\Big]_{-\frac{1}{2}}^{\frac{1}{2}}= \frac{1}{12} \\
%     \E U_i^4 &= \int_{-\frac{1}{2}}^{\frac{1}{2}} x^4 dx = \Big[\frac{1}{5}x^5\Big]_{-\frac{1}{2}}^{\frac{1}{2}}= \frac{1}{80} \\
%     \Var U_i^2 &= \E U_i^4 - \Big(\E U_i^2\Big)^2 = \frac{1}{180}.
% \end{align*}
% \begin{align*}
%     \E U_i^2 = \frac{1}{12},\;\;\; \E U_i^4 = \frac{1}{80},\;\;\; \Var U_i^2 = \frac{1}{180}.
% \end{align*}

\subsection{Distribution matching with perturbation} 
%\todo{Come up with a better name?}
To match the distribution of the rotated dithered quantization to a Gaussian, we can scale the lattice by $s\in\R$ and add an independent noise term $Z\in\R^n$, where $P_{Z_i}=P_{Z_j}, \forall_{i,j}$, at the decoder. The resulting reconstruction follows $Y=X+R(sU+Z)$ while the quantization error is $R(sU+Z)$. By Proposition \ref{proposition:R_makes_rot_invariant}, $R(sU+Z)$ is rotationally invariant, which implies $\normsqr{(R(sU+Z))}=\normsqr{(sU+Z)}$. Therefore, to simplify the analysis and notation, we focus on $\normsqr{(sU+Z)}$:
\begin{align}
        &||(sU+Z)||_2^2=\sum_{i=1}^n (sU+Z)_i^2 \stackrel{\text{approx.}}{\sim}\mathcal{N}\left(\bar{\mu},\bar{\sigma}^2\right),\\ 
        &\text{where }\bar{\mu}=n \Eb{(sU+Z)_1^2} \text{ and } \bar{\sigma}^2=n\Varb{(sU+Z)_1^2}.
        \nonumber
\end{align}
To match it with a Gaussian error distribution we consider the second and fourth moments of the marginals:
\begin{align}
    \Eb{(sU+Z)_i^2} &= \Eb{s^2 U_i^2} + 2\Eb{sU_iZ_i} + \Eb{Z_i^2} \nonumber \\
    &= \frac{s^2}{12} + \Eb{Z_i^2}, \\
    \Eb{(sU+Z)_i^4} &= \Eb{s^4 U_i^4} + 4\Eb{s^3 U_i^3 Z_i} \nonumber
    + 6\Eb{s^2 U_i^2 Z_i^2} + 4\Eb{sU_i Z_i^3} + \Eb{Z_i^4}   \nonumber \\
    &= \frac{s^4}{80} + 6s^2\Eb{U_i^2 Z_i^2} + \Eb{Z_i^4} \label{equation:4th_moment_terms} \\
    &= \frac{s^4}{80} + \frac{s^2\Eb{Z_i^2}}{2} + \Eb{Z_i^4}, \\
    \Varb{(sU+Z)_i^2} &= \frac{s^4}{180} + \frac{s^2\Eb{Z_i^2}}{3} - \Eb{Z_i^2}^2 + \Eb{Z_i^4}, 
\end{align}
since $U_i$ is symmetric around the origin, the expectation of its odd powers vanish, 
and in (\ref{equation:4th_moment_terms}), 
we have $\Eb{(Z_iU_i)^2} = \Eb{Z_i^2} \Eb{U_i^2}$ by the independence of $Z$ and $U$.
% \begin{align}
%     \E U_i^2 Z_i^2 &= \int u^2 z^2 dP_{U_i}(u) dP_{Z_i}(z) \\
%         &= \int \int u^2 z^2 dP_{U_i}(u) dP_{Z_i}(z) \\
%         &= \int u^2 dP_{U_i}\int  z^2 dP_{Z_i}(z) \\
%         &= \E U_i^2 \E Z_i^2 = \frac{\E Z_i^2}{12}.
% \end{align}
Matching the moments of $\normsqr{G}$ in (\ref{eqn:l2_gaussian}) to those of $\normsqr{(sU+Z)_i^2}$ identified above, we get
\begin{align}
    &\begin{cases}
    n &= n \left( \frac{s^2}{12} + \Eb{Z_i^2} \right), \\
    2n &= n \left(\frac{s^4}{180} + \frac{s^2\Eb{Z_i^2}}{3} - \left(\Eb{Z_i^2}\right)^2 + \Eb{Z_i^4} \right)
    \end{cases} 
    \\
    &\text{which results in} \nonumber \\
    &\begin{cases}
    \Eb{Z_i^2} &= 1 - \frac{s^2}{12}, \label{equation:Z_2nd_moment} \\
    \Eb{Z_i^4} &= \frac{7s^4}{240} - \frac{s^2}{2}+3.
    \end{cases}
\end{align}
The scale parameter $s$ determines the size of the quantization bins; the larger the bins are, the fewer bits need to be communicated on average. However, Equation (\ref{equation:Z_2nd_moment})  places an upper bound of $s\leq2\sqrt{3}$ on the scale parameter since the second moment of any distribution is non-negative.

A distribution that achieves this upper bound is the Weibull distribution. Let $Z_i\sim \text{Weibull}(\lambda, k), \lambda>0, k>0, \pr{Z_i\leq t}=1-e^{-(\frac{t}{\lambda})^k}$ for $t>0$, and $0$ otherwise. The moments of Weibull distribution are $\Eb{Z_i^n} = \lambda^n \Gamma(1+\frac{n}{k})$. Then, we have
\begin{align}
    &\begin{cases}
        \Eb{Z_i^2} = \lambda^2 \Gamma(1+\frac{2}{k}) = 1 - \frac{s^2}{12}, \\
        \Eb{Z_i^4} = \lambda^4 \Gamma(1+\frac{4}{k}) = \frac{7s^4}{240} - \frac{s^2}{2}+3.
    \end{cases}
\end{align}
The non-negative solutions for these equations are
\begin{align}
    s(k) &= 2\sqrt{3} \sqrt{ \frac
    {\displaystyle \sqrt{30}\gtwo\sqrt{-3\gtwo^2+\gfour} +15\gtwo^2 - 5\gfour}
    { 21\gtwo^2 - 5\gfour }}, \label{equation:s_for_weibull} \\
    \lambda(k) &= \sqrt{\frac{-\sqrt{30}\sqrt{-3\gtwo^2+\gfour}+6\gtwo}
    { 21\gtwo^2 - 5\gfour }}, \label{equation:l_for_weibull}
\end{align}
where $\gamma_2\triangleq\Gm{1+\frac{2}{k}}$ and $\gamma_4\triangleq\Gm{1+\frac{4}{k}}$. Equations (\ref{equation:s_for_weibull}) and (\ref{equation:l_for_weibull}) are plotted in Figure \ref{figure:weibull_noise_parameters_log}, we observe that, for small values of $k$, $s(k)$ is close to the upper bound $2\sqrt{3}$, which maximizes the size of the quantization bins, minimizing the rate, for which the parameter $\lambda$ approaches zero. 
% Hence, in Figure \ref{figure:weibull_noise_parameters_log} we plot the behaviour of $\lambda(k)$ and $2\sqrt{3}-s(k)$ in logarithmic scale.

% \begin{figure}[t]
%     \centering
%     \includegraphics[width=0.9\columnwidth]{figures/plot_Weibull_params.pdf}
%     \caption{Parameters $\lambda(k)$ and $s(k)$ in Equations (\ref{equation:s_for_weibull}) and (\ref{equation:l_for_weibull}).}
% \label{figure:weibull_noise_parameters}
% \end{figure}

\begin{figure}[t]
    \centering
    \includegraphics[width=0.5\columnwidth]{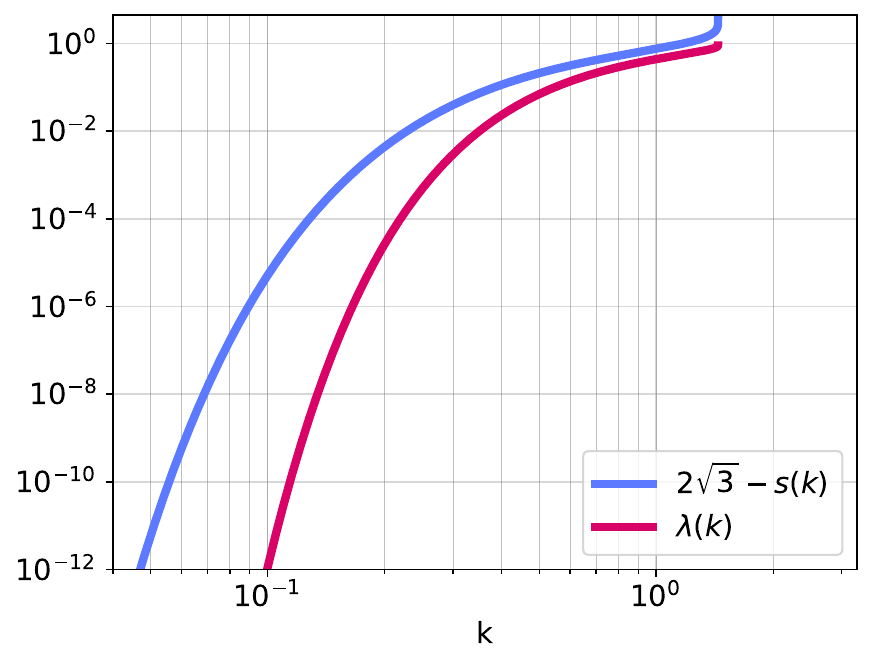}
    \caption{Parameters $\lambda(k)$ and $2\sqrt{3}-s(k)$ in Equations (\ref{equation:s_for_weibull}) and (\ref{equation:l_for_weibull}) with logarithmic scale.}
\label{figure:weibull_noise_parameters_log}
\end{figure}

\section{KL-divergence analysis} \label{section:kl_analysis}
In the previous section, we have chosen the scale parameter and the the noise variable such that the moments of the transformed dithered quantization noise variable match those of the desired Gaussian noise. In this section we compute the KL-divergence between the two, and show it diminishes as $O(n^{-1})$.
\begin{theorem}[Theorem 2 in \cite{moulin_kullback-leibler_2014}]  \label{theorem:kl_sum_convergence}
    Let $X_i\in\R, Y_i\in\R$ be scalar r.v.s with $\text{support}(X_i) \subseteq \text{support}(Y_i)$, with the same means, bounded fourth moments and bounded, continuously differentiable probability density functions $p_X, p_Y$. Let the distribution of the sum of $n$ independent variables $X_i$ and $Y_i$ be denoted by $P_{\sum_n X}$ and $P_{\sum_n Y}$ respectively, and
    $P_X^*=\mathcal{N}\left(0, \Var{X}\right)$, $P_Y^*=\mathcal{N}\left(0, \Var{Y}\right)$, then:
    \begin{equation}
    	\D(P_{\sum_n X}||P_{\sum_n Y}) = \D(P_{X}^*||P_{Y}^*) + O(n^{-1}).
    \end{equation}
\end{theorem}

\begin{proposition} \label{proposition:kl_convergence}
    For $U_i\sim\unifint{-\sfrac{1}{2},\sfrac{1}{2}}, G_i\sim\mathcal{N}(0, I)$, and $Z\sim\text{Weibull}(\lambda,k)$ with $\lambda(k), s(k)$ defined according to equations (\ref{equation:s_for_weibull}), (\ref{equation:l_for_weibull}), we have:
    \begin{equation}
        \D(P_{\normsqr{sU+Z}}||P_{\normsqr{G}})=O(n^{-1}).
    \end{equation}
\end{proposition}

\begin{proof}
    (Outline; details provided in the Appendix.) 
    Through the construction, $s$ and $Z$ were chosen such that $(sU+Z)_i^2$ and $G_i^2$ have equal mean and variance. Thus, $P^*_{(sU+Z)^2}=P^*_{G^2}$, and so $\D(P^*_{(sU+Z)^2} || P^*_{G^2})=0$.
    Let $n \equiv 0 \mod 5$ then summed groups of r.v.s $\sum_{j=1}^5 (sU+Z)_{5k+j}^2$ and $\sum_{j=1}^5 G_{5k+j}^2$ for group index $k\in\{0,\dots,\frac{n}{5}\}$ fulfill the assumptions of Theorem \ref{theorem:kl_sum_convergence}. This follows, $G_i^2$ being a Chi-squared r.v., density of which becomes continuously differentiable when summing more than $5$ variables.
    Applying the theorem yields
     \begin{align*}
        \D(P_{\normsqr{sU+Z}}||P_{\normsqr{G}}) &=\D(P_{\sum_n (sU+Z)_i^2}||P_{\sum_n G_i^2}) \\ &=O\left(\frac{5}{n}\right)=O(n^{-1}).
    \end{align*}
\end{proof}

As the densities of the r.v.s $R(sU+Z), \norm{sU+Z}, \normsqr{sU+Z}$ and $G, \norm{G}, \normsqr{G}$ respectively, differ only by an invertible transformation, their pairwise KL-divergences are the same:
% \begin{equation}
%     \D(P_{\norm{sU+Z}}||P_{\norm{G}}) = \D(P_{\normsqr{sU+Z}}||P_{\normsqr{G}}))=O(n^{-1}),
% \end{equation}
% likewise:
\begin{align}    
    \D(P_{R(sU+Z)}||P_{G}) &= O(n^{-1})\\
    \D(P_{\norm{sU+Z}}||P_{\norm{G}}) &= O(n^{-1}) \\
    \D(P_{\normsqr{sU+Z}}||P_{\normsqr{G}})&=O(n^{-1}).
\end{align}

By the data processing inequality, the KL-divergence between one-dimensional marginal is also bounded by:
\begin{align}
	 \D(P_{(R(sU+Z))_i}||P_{G_i}) \leq \D(P_{R(sU+Z)}||P_{G}) = O(n^{-1}).
\end{align}

\section{Other lattices}
While we have focused on the integer lattice so far, other lattices can also be used with analogous analysis. Let $\lattice$ be an $m$-dimensional lattice.
Given dimension $n$ divisible by $m$, $\frac{n}{m}$ copies of $\lattice$ can be used for dithered quantization, with quantization error $V\in\R^n$, where $V_{(k-1)m+1:km
}\sim\unifset{\cell_\lattice}, k\in\{1, \dots, \frac{n}{m}\}$ is uniformly distributed over the Voronoi cell of $\lattice$. Combined with random rotation $R$, $RV$ is a rotationally invariant distribution as before. We have:
\begin{align}
    \normsqr{V}&=\sum_{k=1}^{n/m} \normsqr{V_{(k-1)m+1:km}}, % \\
    % &\stackrel{\text{approx.}}{\sim} \mathcal{N}\left(\mu=\frac{n}{m}\E \normsqr{V_{:m}}, \sigma^2=\frac{n}{m} \Var \normsqr{V_{:m}}\right), \nonumber 
\end{align}
which is a sum of $\frac{n}{m}$ independent identically distributed r.v.s, and as such, is approximately normal. The mean and variance of $\normsqr{V_{:m}}$ can be obtain from moments of $V$, where, for index vector $a\in\N^n$ the moments are of the form $\E \left[ v_1^{a_1} v_2^{a_2} \dots v_n^{a_n} \right]$:
\begin{align}
    \E \left[ \normsqr{V_{:m}} \right] &= \E \left[\sum_{i=1}^m v^2_i \right] = \sum_{i=1}^m \E \left[v^2_i \right] \\
    \E \left[ \norm{{V_{:m}}}^4 \right] &= \E \left[ \left(\sum_{i=1}^m v^2_i\right)^2 \right] \nonumber
    = \sum_{i=1}^m \E \left[ v^4_i \right] + 2 \sum_{i < j} \E \left[v^2_i v^2_j \right].
\end{align}
For a class of lattices called \emph{root lattices}, these moments can be exactly computed \cite{kohn_moment_2018}, based on dividing the Voronoi lattice into identical simplices called \emph{fundamental simplices} \cite{conway_voronoi_1993}. %p. 531
 % Moment Varieties of Measures on Polytopes
Table \ref{table:lattice_properties} shows $\Eb{\normsqr{V_{:m}}}$ and $ \Varb{\normsqr{V_{:m}}}$ computed for different root lattices. 
For the second moment of the $24$-dimensional leech lattice, we use the numerical result calculated in \cite{conway_voronoi_1984}. We can follow the same idea of random rotation $R$, scaling $s$ and addition of noise $Z$ to achieve desired Gaussian channel $\normsqr{R(sV+Z)}\stackrel{\text{approx.}}{\sim} \mathcal{N}(n, 2n)$.
We do not calculate what is an appropriate noise $Z$ for higher dimensional lattices but it can be done either analytically or through numerical methods.
The KL-analysis directly applies to this case as well.

\begin{table}[h]
    \caption{Properties of different lattices. $A_1=\sqrt{2}\Z$ is a scaled version of integer lattice, $A_2$ is the hexagonal lattice, and $\Lambda_{24}$ is the Leech lattice.}
    \centering
    \def\arraystretch{1.30}
    \begin{tabular}{c|c|c|c|c|r}
    Name &
    \begin{tabular}[x]{@{}c@{}}Dim.\\ $m$ \end{tabular} &
    \begin{tabular}[x]{@{}c@{}}Vol.\\ $\cell_\lattice$ \end{tabular} &
     $\E \normsqr{V_{:m}}$ & $\Var \normsqr{V_{:m}}$ &
    \begin{tabular}[x]{@{}l@{}}Excess\\ info. \\ (Prop. \ref{proposition:overhead_rotated_quantizer})  \\ $\sfrac{\text{bits}}{\text{dim.}}$ \end{tabular} 
    \\ \hline
    $A_1$ & $1$ & $\sqrt{2} $ & $\sfrac{1}{6}      $ & $\sfrac{1}{45}            $ & $0.25461$ \\ \hline
    $A_2$ & $2$ & $\sqrt{3} $ & $\sfrac{5}{36}     $ & $\sfrac{43}{3240}         $ & $0.22686$ \\ \hline
    $A_3$ & $3$ & $2        $ & $\sfrac{1}{8}      $ & $\sfrac{29}{2880}         $ & $0.21376$ \\ \hline
    $A_4$ & $4$ & $\sqrt{5} $ & $\sfrac{7}{60}     $ & $\sfrac{77}{9000}         $ & $0.20709$ \\ \hline
    $D_4$ & $4$ & $2        $ & $\sfrac{13}{120}   $ & $\sfrac{167}{25200}       $ & $0.19387$ \\ \hline
    $D_5$ & $5$ & $2        $ & $\sfrac{1}{10}     $ & $\sfrac{11}{2016}         $ & $0.18613$ \\ \hline
    $D_6$ & $6$ & $2        $ & $\sfrac{2}{21}     $ & $\sfrac{533}{105840}      $ & $0.18427$ \\ \hline
    $D_7$ & $7$ & $2        $ & $\sfrac{31}{336}   $ & $\sfrac{79}{16128}        $ & $0.18518$ \\ \hline
    $D_8$ & $8$ & $2        $ & $\sfrac{13}{144}   $ & $\sfrac{139}{28512}       $ & $0.18735$ \\ \hline
    $E_6$ & $6$ & $\sqrt{3} $ & $\sfrac{5}{56}     $ & $\sfrac{2497}{635040}     $ & $0.17230$ \\ \hline
    $E_7$ & $7$ & $\sqrt{2} $ & $\sfrac{163}{2016} $ & $\sfrac{1727}{580608}     $ & $0.16139$ \\ \hline
    $E_8$ & $8$ & $1        $ & $\sfrac{929}{12960}$ & $\sfrac{457579}{230947200}$ & $0.14597$ \\ \hline
    $\Lambda_{24}$ & $24$ & 1
    & \begin{tabular}[x]{@{}r@{}}$0.065771$\\ $\pm0.000074$ \end{tabular} 
    &
    & \begin{tabular}[x]{@{}r@{}}$0.08389$\\ $\pm0.00081$ \end{tabular} 
    \end{tabular}
    \label{table:lattice_properties}
\end{table}

\section{Excess information analysis}

Next, we would like to characterize the rate of communication required for the proposed channel simulation approach and compare it with the alternative layered quantization scheme. Given that $I(X;Y)$ is a lower bound, we compare the excess information each scheme needs to transmit beyond this lower bound.  

\begin{proposition} %{(Layered quantizer overhead.)} 
\label{proposition:overhead_layered_quantizer}
 The excess information of layered quantization in \cite{agustsson_universally_2020} that employs scale mixture of uniform distributions is bounded above by $0.521$ bits. 
\end{proposition}

\begin{proof}
Let $X$ be a scalar r.v., $U', U \sim \unifint{-\sfrac{1}{2},\sfrac{1}{2}}$, and $S = 2\sigma\sqrt{\Gamma}$ where $\Gamma \sim \text{Gamma}(\sfrac{3}{2}, \sfrac{1}{2})$.
Further, let $K = \lfloor X/S - U'\rceil$.
Then
\begin{align}
    Y& = K S + U' \sim X + SU \sim X + \mathcal{N}(0, \sigma^2).
\end{align}
    We have
    \begin{align}
      H(K \mid U', S)
      &= I(X; Y \mid S) \label{eqn:effectivnes_of_dithering} \\
      &= h(Y \mid S) - h(Y \mid X, S) \\
      % &= h(Y \mid S) - \Eb{\log S} \\
      &\leq h(Y) - h(Y \mid X, S) \\
      &= I(X; Y) + h(Y \mid X) - \Eb{\log S} \\
      &= I(X; Y) + \frac{1}{2} \log \pi + \frac{1 - \psi(\frac{3}{2})}{2 \ln 2} - 1 \label{e:Gasussian} \\
      &\leq I(X; Y) + 0.521,
    \end{align}
    % \begin{align}
    %   H(K \mid U', S)
    %   &= I(X; Y \mid S) \\
    %   &= h(Y \mid S) - h(Y \mid X, S) \\
    %   &= h(Y \mid S) - \Eb{\log S} \\
    %   &\leq h(Y) - \Eb{\log S} \\
    %   &\leq I(X; Y) + h(Y \mid X) - \Eb{\log S} \\
    %   &= I(X; Y) + \frac{1}{2} \log \pi + \frac{1 - \psi(\frac{3}{2})}{2 \ln 2} - 1 \label{e:Gasussian} \\
    %   &\leq I(X; Y) + 0.521,
    % \end{align}
    where (\ref{eqn:effectivnes_of_dithering}) follows Theorem 1 in \cite{zamir_universal_1992},
    and (\ref{e:Gasussian}) follows from the fact that  $SU' \sim \mathcal{N}(0, \sigma^2)$ as shown in \cite{qin2003scalemixture}; and $\psi(\cdot)$ is the digamma function.
\end{proof}

This result is surprising considering that much more involved schemes such as adaptive greedy rejection sampling \cite{flamich2023agrs} result in an excess information of nearly $2$ bits for Gaussian sources. Li and El Gamal \cite{li_strong_2018} provide a lower bound on the excess information for discrete channels.
In Figure~\ref{fig:gaussian}, we numerically evaluate the lower bound for a
Gaussian source and a finely discretized one-dimensional Gaussian channel, and various choices of $\sigma$. 
We find that when $I(X; Y) > 1$, the empirical excess information lower bound is close to $0.5$ bits. That is, the layered quantization is close to optimal for a one-dimensional Gaussian source.

\begin{proposition} \label{proposition:overhead_rotated_quantizer}
Let $X$ be a scalar r.v., and $V',V\sim\unifset{\cell_\lattice}$, where $\cell_\lattice$ is the Voronoi cell of lattice $\Lambda$ used for quantization. Let $K$ denote the index of the quantized lattice bin: 
\begin{align}
    K &= q_\lattice(R^\top S^{-1} X - V'), \mbox{~ and }\\
    Y &= R(S(K+V')+Z)=X+R(SV+Z).
\end{align}
Then, the excess information for the rotated dithered quantization scheme is bounded by
\begin{equation}
    H(K|V',R,S,Z) \leq I(X;Y) + h(G) - h(V) - n \Eb{\log S},
\end{equation}
where $G\sim\mathcal{N}(0, \Var(SV+Z))$.
\end{proposition}
\begin{proof} Follows similarly to the proof of Prop. \ref{proposition:overhead_layered_quantizer}.
\begin{align}
    H(K|V',R,S,Z) %= \\
    % % \nonumber \\ &= H(Y|V',R,S,Z)  - H(Y|X,V',R,S,Z) \\
    % % = I(X;Y|V',R,S,Z) \\
    % % &= I(X;Y|R,S,Z) - I(X;V'|R,S,Z) \nonumber \\ &~~~+ I(X;V'|Y,R,S,Z) \\
    % % % &= I(X;Y|R,S,Z) - 0 + I(X;V'|Y,R,S,Z) \\
    % % &\leq 
    % =I(X;Y|R,S,Z) \\
    % % &= h(Y|R,S,Z) - h(Y|X,R,S,Z) \\
    % &\leq h(Y) - h(Y|X,R,S,Z) \\
    % &= h(Y) - h(SV|S) \\
    &\leq I(X;Y) + h(Y|X) - h(SV|S) \\
    &\leq I(X;Y) + h(G) - h(V) - n \Eb{\log S}.
\end{align}
\end{proof}

The values of the upper bound from Proposition \ref{proposition:overhead_rotated_quantizer} are presented in Table \ref{table:lattice_properties}, where the scale $S$ is taken to be deterministic and follow $\pr{S=\left(\E\normsqr{V_{:m}}\right)^{-1/2}}=1$. 
% Since $-\E \log(S)$ is a convex function this minimizes the upper bound. 
With the same integer lattice as in the layered quantization, we can see that the bound on the excess information reduces to $0.255$ bits per dimension, less than half that of layered quantization. This is because it uses a randomized scale $S$, and as such the convex function $-\Eb{\log S}$ term is necessarily greater than that of a fixed $\pr{S=s}=1$, given we know the upper bound for $\Eb{S}$.

We can also observe that the lower bound on the excess information further diminishes when higher dimensional lattices are employed. It reduces all the way down to $0.084$ bits per dimension for the Leech lattice, which is a more than six-fold reduction, showing the benefits of quantization with higher-dimensional lattices.

\begin{figure}
\centering
\includegraphics[width=0.5\columnwidth]{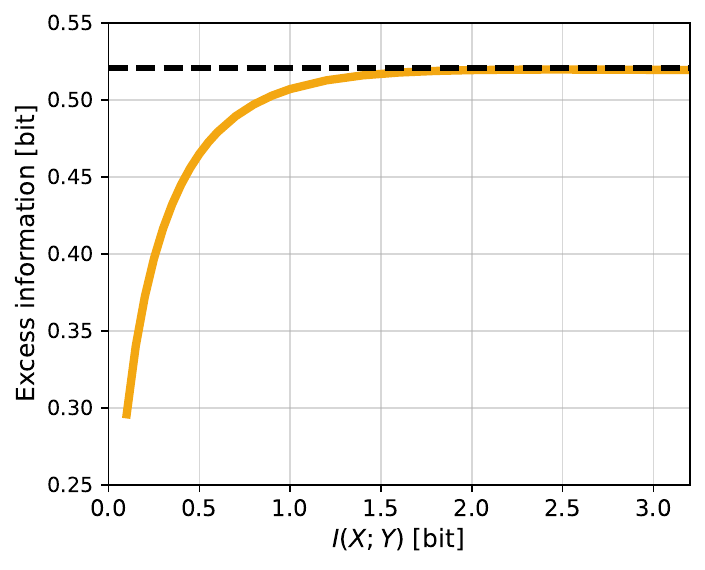}
\caption{
  The dashed line indicates the $0.521$ excess information in the coding cost of the layered
  quantization scheme. The orange line corresponds to a numerical estimate of
  a lower bound on the excess information for a Gaussian source $X \sim
  \mathcal{N}(0, 1)$, a channel $Y \sim \mathcal{N}(X, \sigma^2)$, and various choices of $\sigma$ which translate to different values of $I(X; Y)$ on the x-axis. 
}
\label{fig:gaussian}
\end{figure}

\section{Conclusion}

We proposed rotated dithered quantization as an efficient approach to approximate simulation of multi-variate Gaussian channels from both computation and communication efficiency, leveraging $n$-dimensional dithered quantization. 
Left for further research is the characterisation of noise required for rotated dithering with higher-dimensional lattices, as well as the potential benefits of integrating randomized scaling into the method.

\newpage

\bibliographystyle{ieeetr}
\bibliography{main.bib}

\begin{thebibliography}{10}

\bibitem{flamich_compressing_2020}
G.~Flamich, M.~Havasi, and J.~M. Hernández-Lobato, ``{Compressing {images} by {encoding} {their} {latent} {representations} with {relative} {entropy} {coding}},'' in {\em {Advances in {Neural} {Information} {Processing} {Systems}}}, vol.~33, pp.~16131--16141, Curran Associates, Inc., 2020.

\bibitem{agustsson_universally_2020}
E.~Agustsson and L.~Theis, ``{Universally {Quantized} {Neural} {Compression}},'' in {\em {Advances in {Neural} {Information} {Processing} {Systems}}}, vol.~33, pp.~12367--12376, Curran Associates, Inc., 2020.

\bibitem{theis_lossy_2022}
L.~Theis, T.~Salimans, M.~D. Hoffman, and F.~Mentzer, ``Lossy {Compression} with {Gaussian} {Diffusion},'' Dec. 2022.
\newblock arXiv:2206.08889.

\bibitem{havasi_minimal_2019}
M.~Havasi, R.~Peharz, and J.~M. Hernández-Lobato, ``{Minimal {Random} {Code} {Learning}: {Getting} {Bits} {Back} from {Compressed} {Model} {Parameters}},'' in {\em {7th {International} {Conference} on {Learning} {Representations} {(ICLR)}, {New} {Orleans}, {LA},{USA}}}, 2019.

\bibitem{isik_communication-efficient_2023}
B.~Isik, F.~Pase, D.~Gunduz, S.~Koyejo, T.~Weissman, and M.~Zorzi, ``Adaptive compression in federated learning via side information,'' in {\em The 27th International Conference on Artificial Intelligence and Statistics (AISTATS)}, May 2024.

\bibitem{triastcyn_dp-rec_2021}
A.~Triastcyn, M.~Reisser, and C.~Louizos, ``{DP}-{REC}: {Private} \& {Communication}-{Efficient} {Federated} {Learning},'' Dec. 2021.

\bibitem{pmlr-v151-shah22b}
A.~Shah, W.-N. Chen, J.~Ball\'e, P.~Kairouz, and L.~Theis, ``Optimal compression of locally differentially private mechanisms,'' in {\em Proceedings of The 25th International Conference on Artificial Intelligence and Statistics} (G.~Camps-Valls, F.~J.~R. Ruiz, and I.~Valera, eds.), vol.~151 of {\em Proceedings of Machine Learning Research}, pp.~7680--7723, PMLR, 28--30 Mar 2022.

\bibitem{burak_comm_2024}
B.~Hasircioglu and D.~Gunduz, ``Communication efficient private federated learning using dithering,'' in {\em IEEE International Conference on Acoustics, Speech and Signal Processing}, Apr. 2024.

\bibitem{cuff_communication_2008}
P.~Cuff, ``Communication requirements for generating correlated random variables,'' in {\em 2008 {IEEE} {International} {Symposium} on {Information} {Theory}}, pp.~1393--1397, July 2008.
\newblock ISSN: 2157-8117.

\bibitem{bennett_entanglement-assisted_2002}
C.~Bennett, P.~Shor, J.~Smolin, and A.~Thapliyal, ``{Entanglement-assisted capacity of a quantum channel and the reverse {Shannon} theorem},'' {\em IEEE Transactions on Information Theory}, vol.~48, pp.~2637--2655, Oct. 2002.

\bibitem{winter_compression_2002}
A.~Winter, ``{Compression of sources of probability distributions and density operators},'' {\em arXiv:quant-ph/0208131}, Aug. 2002.

\bibitem{li_strong_2018}
C.~T. Li and A.~El~Gamal, ``{Strong {Functional} {Representation} {Lemma} and {Applications} to {Coding} {Theorems}},'' {\em IEEE Transactions on Information Theory}, vol.~64, pp.~6967--6978, Nov. 2018.

\bibitem{maddison2016ppmc}
C.~J. Maddison, ``{A Poisson process model for Monte Carlo},'' in {\em Perturbation, Optimization, and Statistics} (T.~Hazan, G.~Papandreou, and D.~Tarlow, eds.), MIT Press, 2016.

\bibitem{harsha_communication_2010}
P.~Harsha, R.~Jain, D.~McAllester, and J.~Radhakrishnan, ``{The {Communication} {Complexity} of {Correlation}},'' {\em IEEE Transactions on Information Theory}, vol.~56, pp.~438--449, Jan. 2010.

\bibitem{theis_algorithms_2022}
L.~Theis and N.~Y. Ahmed, ``{Algorithms for the {Communication} of {Samples}},'' in {\em {Proceedings of the 39th {International} {Conference} on {Machine} {Learning}}}, pp.~21308--21328, PMLR, June 2022.

\bibitem{ziv_universal_quantization}
J.~Ziv, ``On universal quantization,'' {\em IEEE Transactions on Information Theory}, vol.~31, no.~3, pp.~344--347, 1985.

\bibitem{zamir_universal_1992}
R.~Zamir and M.~Feder, ``On universal quantization by randomized uniform/lattice quantizers,'' {\em IEEE Transactions on Information Theory}, vol.~38, pp.~428--436, Mar. 1992.
\newblock Conference Name: IEEE Transactions on Information Theory.

\bibitem{CHOY2003227}
S.~Choy and S.~G. Walker, ``The extended exponential power distribution and {B}ayesian robustness,'' {\em Statistics \& Probability Letters}, vol.~65, no.~3, pp.~227--232, 2003.

\bibitem{hegazy_randomized_2022}
M.~Hegazy and C.~T. Li, ``Randomized {Quantization} with {Exact} {Error} {Distribution},'' in {\em 2022 {IEEE} {Information} {Theory} {Workshop} ({ITW})}, pp.~350--355, Nov. 2022.

\bibitem{FUNG20082883}
T.~Fung and E.~Seneta, ``A characterisation of scale mixtures of the uniform distribution,'' {\em Statistics \& Probability Letters}, vol.~78, no.~17, pp.~2883--2888, 2008.

\bibitem{Ling:ITW:23}
C.~W. Ling and C.~Ting~Li, ``Vector quantization with error uniformly distributed over an arbitrary set,'' in {\em 2023 IEEE International Symposium on Information Theory (ISIT)}, pp.~856--861, 2023.

\bibitem{moulin_kullback-leibler_2014}
P.~Moulin and P.~R. Johnstone, ``Kullback-{Leibler} {Divergence} and the {Central} {Limit} {Theorem},'' in {\em UCSD Information Theory and Applications Workshop}, 2014.

\bibitem{kohn_moment_2018}
K.~Kohn, B.~Shapiro, and B.~Sturmfels, ``Moment {Varieties} of {Measures} on {Polytopes},'' {\em Annali della Scuola Normale Superiore di Pisa, Classe di Scienze}, vol.~21, pp.~739--770, 2020.

\bibitem{conway_voronoi_1993}
J.~H. Conway and N.~J.~A. Sloane, ``Voronoi {Cells} of {Lattices} and {Quantization} {Errors},'' in {\em Sphere {Packings}, {Lattices} and {Groups}} (J.~H. Conway and N.~J.~A. Sloane, eds.), Grundlehren der mathematischen {Wissenschaften}, pp.~449--475, New York, NY: Springer, 1993.

\bibitem{conway_voronoi_1984}
J.~H. Conway and N.~J.~A. Sloane, ``On the {Voronoi} {Regions} of {Certain} {Lattices},'' {\em SIAM Journal on Algebraic Discrete Methods}, vol.~5, pp.~294--305, Sept. 1984.
\newblock Publisher: Society for Industrial and Applied Mathematics.

\bibitem{qin2003scalemixture}
Z.~Qin, P.~Damien, and S.~Walker, ``{Uniform Scale Mixture Models With Applications to {B}ayesian Inference},'' in {\em AIP Conference Proceedings}, 2003.

\bibitem{flamich2023agrs}
G.~Flamich and L.~Theis, ``Adaptive greedy rejection sampling,'' in {\em IEEE International Symposium on Information Theory}, 2023.

\end{thebibliography}

\clearpage
\onecolumn

\appendix
In this appendix, we provide more details of the proof of Proposition \ref{proposition:kl_convergence}, which shows the $\ell_2$ norm of a $n$-dimensional Gaussian noise $G$ and rotated dithered quantization noise $(sU+Z)$ converge in KL-divergence on the order of:
\begin{equation}
    \D(P_{\normsqr{sU+Z}}||P_{\normsqr{G}})=O(n^{-1}).
\end{equation}
It follows the application of Theorem \ref{theorem:kl_sum_convergence} which requires the densities of the random variables to be bounded, and continuously differentiable. 
The theorem, can not be applied directly to $G^2_i$ and $R(sU+Z)^2_i$ since their densities are unbounded. Instead, we apply the theorem to sums of 5 random variables where we interpret the sum of $n$ random variables as a sum of $n/5$ random variables:
\begin{align}
    \normsqr{G} &= \sum_{i=1}^n G_i^2 = \sum_{i=1}^{n/5} \sum_{j=1}^5 G_{5(i-1)+j}^2 \\
    \normsqr{R(sU+Z)} = \normsqr{sU+Z} &= \sum_{i=1}^n (sU+Z)_i^2 = \sum_{i=1}^{n/5} \sum_{j=1}^5 (sU+Z)_{5(i-1)+j^2}.
\end{align}

\begin{proposition} \label{propositon:G5_is_cont_diff}
    Let $G_i\sim\mathcal{N}(0,1)$ for $i\in\{1,\dots,5\}$ be independently distributed, then the probability density function of $\sum_{i=1}^{5} G_i^2$ is bounded and continuously differentiable.
\end{proposition}

\begin{proof}
  The random variable $\sum_{i=1}^{5} G_i^2$ is chi-squared distributed with density function:
    \begin{equation}
        p_{\sum_{i=1}^5 G_i^2}(x) = 
        \begin{cases}
            0 & \text{if } x < 0, \\
            \frac{2^{5/2}}{\Gamma(5/2)} x^{3/2}e^{-x/2} & \text{if } x \geq 0,
        \end{cases}
    \end{equation}
    which is a bounded and continuously differentiable.
\end{proof}

\begin{table}[b]
\centering
\def\arraystretch{1.20}
\caption{Probability density functions of sums of squared normal random variables and their properties. $C_i$ denote different constants.}
    \begin{tabular}{c|l|c|c|c|l}
    $N$ &
      \begin{tabular}[c]{@{}l@{}}density of $\sum_{i=1}^N G_i^2$\\ for $x>0$\end{tabular} &
      bound. &
      cont. &
      \begin{tabular}[c]{@{}c@{}}cont.\\ diff.\end{tabular} &
      comment \\ \hline
    $1$ & $C_1 e^{-x/2} x^{-1/2}$ & - & - & - & $\lim_{x \to 0^+} p \to \infty$             \\ \hline
    $2$ & $C_2 e^{-x/2}$          & + & - & - & jump discontinuity at $x=0$                 \\ \hline
    $3$ & $C_3 e^{-x/2} x^{1/2}$  & + & + & - & $\lim_{x \to 0^+} \frac{dp}{dx} \to \infty$ \\ \hline
    $4$ & $C_4 e^{-x/2} x$        & + & + & - & jump discontinuity of derivative at $x=0$   \\ \hline
    $5$ & $C_5 e^{-x/2} x^{3/2}$  & + & + & + &
    \end{tabular}
\label{table:G_convolutions}
\end{table}

To motivate better, the choice of considering the random variables in the groups of $5$, we show the properties of sums of $G_i$ in table \ref{table:G_convolutions}. We can see that the more terms we sum, the more 'regular' the probability density function. For $N$ terms the density is proportional to $x^{N/2-1}$ for $x>0$ and $0$ otherwise, which becomes continuous only when the exponent of the polynomial is positive i.e., $N\geq3$. Similarly, the derivative of the density is proportional to $x^{N/2-3}$ which is continuous for $N\geq5$.

\begin{proposition} \label{propositon:f5_is_cont_diff}
    Let $s>0, U_i\sim\unifint{-\sfrac{1}{2}, \sfrac{1}{2}}$ and $Z_i$ be Weibull distributed with parameters $\lambda>0, k>0$ for $i\in\{1,\dots,5\}$ be independently distributed. The probability density function of $\sum_{i=1}^{5} (sU+Z)_i^2$ is bounded, continuously differentiable.
\end{proposition}

The density of the sum of independent random variables is the convolution of their respective densities, thus first we introduce a few of their properties. 
\begin{definition} \label{definition:convolution}
    For functions $g,h: \R \to\R$, let their convolution be:
    \begin{align}
        (g*h)(x) = \int^\infty_{-\infty} g(t) h(x-t) dt.
    \end{align}
\end{definition}
Let $L^p, p>0$ denote a space of measurable functions such that $g\in L^p \Rightarrow \left(\int |g|^p d\mu\right)^{\sfrac{1}{p}} < \infty$.
\begin{theorem}\label{theorem:youngs_convolution} [Young's convolution inequality]
    Let $g\in L^p, h\in L^q$ and $\frac{1}{p}+\frac{1}{q}=\frac{1}{r}+1$ then $g*p\in L^r$.
\end{theorem}
\begin{corollary} \label{corollary:l_1_convolution}
    Let $g\in L^1, h\in L^1$ then $g*h\in L^1$.
\end{corollary}
\begin{corollary} \label{corollary:l_infty_convolution}
    Let $g\in L^1, h\in L^\infty$ then $g*h\in L^\infty$.
\end{corollary}
\begin{theorem} \label{theorem:convolution_continutiy}
    Let $g\in L^1, h\in L^\infty$ then $g*h$ is uniformly continuous.
\end{theorem}
% This holds for all 1/p+1/q=1, L^p and L^q.

\proof{
(Proposition \ref{propositon:f5_is_cont_diff})
Unlike for the sums of squared Gaussian random variables $G_i^2$, we can not explicitly compute the density of 
% sums of squared uniform and Weibull variables 
$(sU+Z)_i^2$. Thus, showing that the sums of 5 such independent random variables require more careful analysis. We will proceed by first calculating the density $f=p_{(sU+Z)^2}$ explicitly. We will show that the density of $\sum_{i=1}^3 (sU+Z)_i^2$ is continuous, and finally we show density of $\sum_{i=1}^5 (sU+Z)_i^2$ is continuously differentiable. The density of $(sU+Z)_i$ is:
\begin{align}
    p_{(sU+Z)_i}(x) 
    &= \int_\R p_{Z_i}(t) p_{sU_i}(x-t) dt 
    = \frac{1}{s} \int_{x-\frac{s}{2}}^{x+\frac{s}{2}} p_{Z_i}(t) dt 
    = \frac{1}{s}\left(\pr{Z_i\leq x+\frac{s}{2}} - \pr{Z_i\leq x-\frac{s}{2}}\right) \\
    &=
    \begin{cases}
            0 & \text{if } x \leq -\frac{s}{2}, \\
            \frac{1}{s}\left[1 - \exp{\left(-(\frac{x+\frac{s}{2}}{\lambda})^k\right)}\right]
                & \text{if } -\frac{s}{2} < x \leq \frac{s}{2}, \\
            \frac{1}{s}\left[\exp{\left(-(\frac{x-\frac{s}{2}}{\lambda})^k\right)} - \exp{\left(-(\frac{x+\frac{s}{2}}{\lambda})^k\right)} \right]
                & \text{if } \frac{s}{2} < x.
    \end{cases} \nonumber
\end{align}
which follows from the cumulative density function of Weibull distribution $\pr{Z \leq x} = 1- \exp{\left(-(\frac{x}{\lambda})^k\right)}$.
The density of $(sU+Z)_i^2$, denoted $f$ for convenience, is:
\begin{align}
    f = p_{(sU+Z)_i^2}(x) =     
    \begin{cases}
            0 & \text{if } x \leq 0, \\
            \frac{1}{s\sqrt{x}}\left[2 - \exp{\left(-(\frac{\sqrt{x}+\frac{s}{2}}{\lambda})^k\right)} - \exp{\left(-(\frac{-\sqrt{x}+\frac{s}{2}}{\lambda})^k\right)}\right]
                & \text{if } 0 < x \leq \frac{s^2}{4}, \\
            \frac{1}{s\sqrt{x}}\left[\exp{\left(-(\frac{\sqrt{x}-\frac{s}{2}}{\lambda})^k\right)} - \exp{\left(-(\frac{\sqrt{x}+\frac{s}{2}}{\lambda})^k\right)} \right]
                & \text{if } \frac{s^2}{4} < x.
    \end{cases} \label{eqn:f_density}
\end{align}
We denote the density of sum of $N$ variables, which is the $N$-fold convolution of the density function $f$, as:
\begin{equation}
    p_{\sum_{i=1}^N (sU+Z)_i^2}(x) = \underbrace{f*f*\ldots*f}_{N\text{-times}} = f^{*N}.
\end{equation}
First, let us note that for all $N\in\N, f^{*N}\in L^1$ which is obvious considering they are density functions (also follows Corollary \ref{corollary:l_1_convolution}).
To show $f^{*2}$ is bounded, note that 
\begin{align}
    f\leq\frac{2u(x)}{s\sqrt{x}}, \text{ where } u(x)=
    \begin{cases}
        0 & \text{ if } x \leq 0,\\ 
        1 & \text{ if } 0 < x.\\ 
    \end{cases}
\end{align} 
Then, since both functions are positive 
\begin{align}
f^{*2} \leq \left(\frac{2u(x)}{s\sqrt{x}}\right)^{*2}=\frac{4\pi u(x)}{s^2} \leq \frac{4\pi}{s^2} < \infty,    
\end{align}
combined with $f^{*2}\geq0$, it implies $f^{*2}\in L^{\infty}$. 
Thus, since $f\in L^1, f^{*2}\in L^\infty$:
\begin{equation}
    f^{*2}*f = f^{*3}\in L^\infty \text{ and } f^{*3} \text { is continuous},
\end{equation}
by Corollary \ref{corollary:l_infty_convolution} and Theorem \ref{theorem:convolution_continutiy}.
Applying this reasoning inductively, for all $N\geq 3$, $f^{*N}$ is bounded (in $L^\infty$) and continuous - in particular, so is $f^{*5}$.

In general, the derivative of the convolution is the convolution of the derivative $D (g*h) = (Dg) * h$. The problem of applying this to $f$ directly is that $Df \notin L^1$, since it has a singularity of the order $x^{-\sfrac{3}{2}}$ at the origin. Thus, we will proceed by separating $f$ into a normalized part $\fnice$, and the singularity. We shall then show that $D\fnice \in L^1$ and we will explicitly calculate the contribution of the singularity.

Let $g(x)=\sqrt{x} f(x)$, $\cnst = \lim_{x\to0^+} g(t) = 2s^{-1} \left[1 - \exp{\left(-(\sfrac{s}{2\lambda})^k\right)}\right]$
and $\fnice = f - \frac{u(t)C}{\sqrt{x}}$. We can express $\fnice$ around $x>0$ as:
\begin{align}
    \fnice &= f - \frac{C}{\sqrt{x}} = \frac{g(x)}{\sqrt{x}} - \frac{C}{\sqrt{x}} \\
    &= \frac{\lim_{t\to 0^+} g(t) + x \lim_{t\to 0^+} Dg(t) + \frac{x^2}{2} \lim_{t\to 0^+}D^2 g(t) + O(x^3)}{\sqrt{x}} - \frac{C}{\sqrt{x}} \\
    &= x^{\sfrac{1}{2}} (\lim_{t\to 0^+} Dg(t)) + \frac{x^{\sfrac{3}{2}}}{2} (\lim_{t\to 0^+} D^2g(t)) + O\left(x^{\sfrac{5}{2}}\right).
\end{align}
All the $\lim_{t\to 0^+}D^N g(t)$ terms are just the $N$-th derivatives of $s^{-1}\left[\exp{\left(-(\frac{\sqrt{x}-\frac{s}{2}}{\lambda})^k\right)} - \exp{\left(-(\frac{\sqrt{x}+\frac{s}{2}}{\lambda})^k\right)} \right]$ and are readily computable. 
The derivative  $D\fnice, x>0$ around $0$ is proportional to $x^{\sfrac{-1}{2}}$ which is locally integrable. 
The derivative $Df$ has another discontinuity if Weibull parameter $k<1$ around $x=\frac{s^2}{4}$ of the form $h(x)\left(x-\frac{s^2}{4}\right)^{-1+k}$, where $h(x)$ is a continuous, bounded function around $\frac{s^2}{4}$. Since $k>0$, the exponent $-1+k>-1$, and thus, the singularity is integrable. Therefore, $Df$ is locally integrable around $\frac{s^2}{4}$ and so is $D\fnice$. 
The right tail $x>\frac{s^2}{4}+1$ of $Df$ can be computed explicitly from $f$ (Equation \ref{eqn:f_density}) and its absolute value is bounded by an integrable function proportional to $u(x)e^{x^{-\sfrac{k}{2}}}$.
The derivative of the term $\frac{Cu(x)}{\sqrt{x}}$ is $Cu(x)x^{\sfrac{-3}{2}}$ which is integrable on $(a,\infty)$ for $a>0$. Thus, $\fnice = f - \frac{Cu(x)}{\sqrt{x}}$ has an integrable right tail.
Therefore, $D\fnice \in L^1$ since:
\begin{align}
    \int_{-\infty}^\infty |D\fnice(x)| dx = 
    \int_{-\infty}^0 |D\fnice(x)| dx +
    \int_{0}^{\frac{s^2}{4}+1} |D\fnice(x)| dx +
    \int_{\frac{s^2}{4}+1}^\infty |D\fnice(x)| dx < \infty,
\end{align}
where integrability of the first integral follows $\fnice=0$ for $x<0$, the middle one from integrability around $x=0$, $x=\frac{s^2}{4}$ and boundedness otherwise, while the last integral follows integrability of the right tail.

Now, we can show $f^{*5}$ is continuously differentiable:
\begin{align}
    D f^{*5} &= D (f^{*3} * f^{*2}) \\
    &= D\left( f^{*3} * \left(\fnice + \frac{Cu(x)}{\sqrt{x}}\right)^{*2} \right) \\
    &= D\left( f^{*3} * \left(\fnice * \fnice + 2\fnice * \frac{Cu(x)}{\sqrt{x}} + \frac{Cu(x)}{\sqrt{x}} * \frac{Cu(x)}{\sqrt{x}} \right) \right) \\
    &= D\left( f^{*3} * \fnice * \left( \fnice + 2\frac{Cu(x)}{\sqrt{x}} \right) \right)  
    + D\left( f^{*3} * \frac{Cu(x)}{\sqrt{x}} * \frac{Cu(x)}{\sqrt{x}} \right)\\
    &= f^{*3} * \left(D\fnice\right) * \left( \fnice + 2\frac{Cu(x)}{\sqrt{x}} \right)  
    + D\left( f^{*3} * C\pi u(x) \right)\\
    % &= f^{*3} * \left(D\fnice\right) * \left( \fnice + 2\frac{Cu(x)}{\sqrt{x}} \right)  
    % +  f^{*3} * C\pi \left( D u(x) \right) \\
    &= f^{*3} * \left(D\fnice\right) * \left( f + \frac{Cu(x)}{\sqrt{x}} \right)  
    +  f^{*3} * C\pi \delta \\
    &= f^{*3} * \left(D\fnice\right) * \left( f + \frac{Cu(x)}{\sqrt{x}} \right)  
    + C\pi f^{*3},
\end{align}
where the derivative of unit step function $u(x)$ is the delta distribution $\delta$. Then,
\begin{align}
    &f^{*3}\in L^1, D\fnice \in L^1 \Rightarrow f^{*3} * D\fnice \in L_1, \\
    &f^{*3}\in L^\infty, D\fnice \in L^1 \Rightarrow f^{*3} * D\fnice \in L_\infty.
\end{align}
We can split the term:
\begin{align}
    d = f + \frac{Cu(x)}{\sqrt{x}} =
    f + \frac{C\left(u(x)-u(x-1)\right)}{\sqrt{x}} + \frac{Cu(x-1)}{\sqrt{x}} = d_L + d_R
\end{align}
where
\begin{align}
    f + \frac{C\left(u(x)-u(x-1)\right)}{\sqrt{x}} = d_L &\in L^1, \\
    \frac{Cu(x-1)}{\sqrt{x}} = d_R &\in L^\infty. \\
\end{align}
Then,
\begin{align}
    f^{*3} * \left(D\fnice\right) \in L^\infty,  d_L \in L^1 &\implies f^{*3} * \left(D\fnice\right) * d_L \in L^\infty \text{ and continuous}, \\
    f^{*3} * \left(D\fnice\right) \in L^1,  dR \in L^\infty &\implies f^{*3} * \left(D\fnice\right) * d_R \in L^\infty \text{ and continuous},  \\
\end{align}
\begin{align}
    &\implies f^{*3} * \left(D\fnice\right) * d \in L^\infty \text{ and continuous}, \\
    &\implies f^{*3} * \left(D\fnice\right) * d + C\pi f^{*3} = D(f^{*5})\in L^\infty \text{ and continuous.}
\end{align}

}

\end{document}